\documentclass[11pt]{article}
\usepackage{amsmath}
\usepackage{amsthm}
\usepackage{titlesec}
\usepackage{indentfirst}
\theoremstyle{definition}
\theoremstyle{plain}\newtheorem{Th}{Theorem}
\theoremstyle{definition}
\theoremstyle{definition}
\theoremstyle{plain}
\theoremstyle{plain}\newtheorem{Co}[Th]{Corollary}
\theoremstyle{plain}\newtheorem{Lm}[Th]{Lemma} \textwidth 155mm
\begin{document}
\title{{\bf Two-tape finite automata with quantum and classical states}}

\author{Shenggen Zheng$^{a,}$\thanks{{\it  E-mail
address:} zhengshenggen@gmail.com},\hskip 2mm Lvzhou Li$^{a,}$\thanks{{\it  E-mail address:} lilvzhou@gmail.com}, \hskip
2mm Daowen Qiu$^{a,b,c,}$\thanks{Corresponding author. {\it E-mail address:}
issqdw@mail.sysu.edu.cn (D. Qiu).} \\
\small{{\it $^a$Department of
Computer Science, Sun Yat-sen University, Guangzhou 510006, China}}\\
\small{{\it $^b$SQIG--Instituto de Telecomunica\c{c}\~{o}es, IST,
TULisbon, }}\\
\small{{\it  Av. Rovisco Pais 1049-001, Lisbon, Portugal}}\\
{\small {\it $^{c}$The State Key Laboratory of Computer Science,
Institute of Software,}}\\
 {\small {\it Chinese  Academy of Sciences,
 Beijing 100080, China}}}

\date{ }
\maketitle \vskip 2mm \noindent {\bf Abstract}
\par
{\it Two-way finite automata with quantum and classical states}
(2QCFA) were introduced by Ambainis and Watrous, and {\it two-way
two-tape deterministic finite automata} (2TFA) were introduced by
Rabin and Scott. In this paper we study 2TFA and propose a new
computing model called {\it two-way two-tape finite automata with
quantum and classical states} (2TQCFA). First, we give efficient
2TFA algorithms for recognizing languages which can be recognized
by 2QCFA. Second, we give efficient 2TQCFA algorithms to recognize
several languages whose status vis-a-vis 2QCFA have been posed as
open questions, such as $L_{square}=\{a^{n}b^{n^{2}}\mid n\in
\mathbf{N}\}$. Third, we show that $\{a^{n}b^{n^{k}}\mid n\in
\mathbf{N}\}$ can be recognized by {\it $(k+1)$-tape deterministic
finite automata} ($(k+1)$TFA). Finally, we introduce {\it $k$-tape
automata with quantum and classical states} ($k$TQCFA) and prove
that $\{a^{n}b^{n^{k}}\mid n\in \mathbf{N}\}$ can be recognized by
$k$TQCFA.


\vskip 2mm

\section{Introduction}
Interest in quantum computation has steadily increased since
Shor's quantum algorithm for factoring integers in polynomial time
\cite{Shor} and Grover's algorithm of searching in database of
size $n$ with only $O(\sqrt{n})$ accesses \cite{Grover}. As we
know, these algorithms are based on {\it quantum Turing machines}
which are complicated to implement using today's experiment
technology. Therefore, it is natural to consider  much more
restricted quantum computing models.

Classically, as one of the simplest computing models, {\it
deterministic finite automata} (DFA) and {\it nondeterministic
finite automata} (NFA) have been deeply studied \cite{Hop}.
Furthermore, {\it two-way two-tape deterministic finite automata}
(2TFA) and {\it multi-tape finite automata} ($m$TFA) were first
introduced by Rabin and Scott in a seminal paper in 1959
\cite{RS}. Unlike DFA or NFA, 2TFA can recognize not only {\it
regular languages} (RL) but also some {\it context-free languages}
(CFL), and even some  {\it no-context-free languages} (NCFL). In
this paper, we will show that some {\it context-free languages},
and even some {\it no-context-free languages}  can be recognized
by 2TFA in linear time.

Correspondingly, it may be interesting to consider restricted {\it
quantum Turing machines}, such as {\it quantum finite automata}
(QFA).  QFA were first introduced independently by Kondacs and
Watrous \cite{Kon97}, as well as Moore and Crutchfield
\cite{Moore}. As a quantum variant of FA, QFA have attracted wide
attentions in the academic community
\cite{Amb06,Amb-F,Nay,Bro,LiQ2,LiQ3,Qiu3,Qiu4}. Many kinds of QFA
have been proposed and studied(e.g., see \cite{LiQ4}). QFA are
mainly divided into two kinds: {\it one-way quantum finite
automata} (1QFA) and {\it two-way quantum finite automata} (2QFA).
If we compare 1QFA with their classical counterparts DFA, 1QFA
have some weaknesses since they can recognize only a proper subset
of RL \cite{Amb-F}. 2QFA, however, are more powerful than their
classical counterparts {\it two-way deterministic finite automata}
(2DFA) which recognize only RL. 2QFA were first introduced by
Kondacs and Watrous, and then they proved that
$L_{eq}=\{a^{n}b^{n}\mid n\in \mathbf{N}\}$ and
$L_{trieq}=\{a^{n}b^{n}c^{n}\mid n\in \mathbf{N}\}$ can be
recognized by 2QFA \cite{Kon97}, where the first language is CFL,
and the second one is NCFL. However, 2QFA have  a disadvantage in
the sense that we need at least $\mathbf{O}(\log n)$ qubits to
store the position of the tape head, which is relative to the
length of the input. In order to conquer this disadvantage,
Ambainis and  Watrous proposed {\it two-way finite automata with
quantum and classical states} (2QCFA)\cite{AJ}, an intermediate
model between 1QFA and 2QFA, which are still more powerful than
their classical counterparts 2DFA. A 2QCFA is essentially a
classical 2DFA augmented with a quantum component of constant
size, where the dimension of the associated Hilbert space does not
depend on the length of the input.  Ambainis and  Watrous showed
$L_{pal}=\{\omega\in\{a,b\}^{*}\mid\omega=\omega^{R}\}$ and
$L_{eq}=\{a^{n}b^{n}\mid n\in \mathbf{N}\}$ can be recognized by
2QCFA \cite{AJ}. Several other no-regular languages including
$M_{eq}=\{a^{n}b_{1}^{n}a^{m}b_{2}^{m}\mid n,m\in \mathbf{N}\}$,
$L_{eq}(k,a)=\{a^{kn}b^{n}\mid n\in \mathbf{N}\}$ and
$L_{=}=\{\omega\in \{a,b\}^{*}\mid\#_{\omega}(a)=\#_{\omega}(b)\}$
were proved to be recognized by 2QCFA in Qiu's paper \cite{Qiu2}.

In this paper we  revisit {\it two-tape two-way deterministic
finite automata} (2TFA), of which the definition has a little
difference from the original model introduced  by Rabin and Scott
in 1959 \cite{RS}. When we mention 2TFA in the following, we refer
to the new defined model without ambiguity. First, we prove that
all the languages mentioned above that can be recognized by 2QFA
or 2QCFA, including $L_{pal}$, $L_{eq}$, $L_{trieq}$, $M_{eq}$,
$L_{eq}(k,a)$ and $L_{=}$,  can also be recognized by 2TFA in
linear time. Note that 2QCFA need polynomial or exponential time
to recognize these languages. Second,  We show some languages,
including $L_{copy}=\{\omega\omega\mid\omega\in\Sigma^{*}\}$,
$L_{middle}=\{xay\mid x,y\in\Sigma^{*},a\in\Sigma\}$, and
$L_{balanced}=\{x\in\{``(",``)"\}^*\mid$ parentheses in $x$ are
balanced $\}$,  can be recognized by 2TFA in linear time. However,
as far as we know, it is still not clear whether 2QFA or  2QCFA
can recognize these languages or not. Thus, it seems that 2TFA are
more powerful than 2QFA or 2QCFA.

Based on 2TFA, we propose 2TQCFA which are similar to 2QCFA, only
augmented with another tape. 2TQCFA  can be implemented with a
quantum part of constant size, and they are shown to be powerful.
We consider the language: $L_{square}=\{a^{n}b^{n^{2}}\mid n\in
\mathbf{N}\}$, whether it can be recognized by 2QFCA or not is
still an open problem proposed by  Ambainis and Watrous in
\cite{AJ}. Rao proposed a new model of automata called {\it
Two-Way Optical Interference Automata} (2OIA) which use the
phenomenon of interference, and proved that they can recognize
$L_{square}$ \cite{MV}. However, 2OIA are complicated to
implement, because it is hard to control  the phenomenon of
interference in high precision. In this paper, we show that
$L_{square}$ can be recognized by 2TQCFA in polynomial time.

Also, we consider {\it multi-tape deterministic finite automata}
($m$TFA). We show that $L_{kpow}=\{a^{n}b^{n^{k}}\mid n\in
\mathbf{N}\}$ can be recognized by {\it $(k+1)$-tape deterministic
finite automata} ($(k+1)$TFA) in linear time. Furthermore, we
propose a new model of quantum automata, {\it multi-tape finite
automata with quantum and classical states} ($m$TQCFA), which is
essentially the model of
 $m$TFA augmented with a quantum component of constant
size. We prove that $L_{kpow}$ can be recognized by {\it $k$-tape
finite automata with quantum and classical states} ($k$TQCFA) in
polynomial time.

The remainder of this paper has the following organization. In
Section 2 we give the definitions of 2TFA, 2TQCFA, $m$TFA and
$m$TQCFA. Then, in Section 3 we show 2TFA can recognize those
languages mentioned before. Afterwards, in Section 4 we describe a
2TQCFA  for recognizing $L_{square}$. Subsequently, in Section 5
we describe a $(k+1)$TFA and a $k$TQCFA for recognizing
$L_{kpow}$. Finally we make a conclusion in Section 6.

\section{Definitions}
\subsection{Definition of 2TFA}
The definition here is slightly different from the original model
introduced by Rabin and Scott \cite{RS}. A 2TFA is defined by a
6-tuple
\begin{equation}
\mathcal{M}=(S,\Sigma,\delta,s_{0},S_{acc},S_{rej})
\end{equation}
where,
\begin{itemize}
\item $S$ is a finite set of classical states;

\item $\Sigma$ is a finite set of input symbols; the tape symbol
set is $\Gamma=\Sigma\cup\{\ |\hspace{-2.0mm}C,\$\}$; two tapes
with the same input $\omega$ are $\mathcal{T}_{1}=
|\hspace{-2.0mm}C_{1}\omega\$_{1}$ and
$\mathcal{T}_{2}=|\hspace{-2.0mm}C_{2}\omega\$_{2}$; we assume
that the input in $\mathcal{T}_{i}$ is started by the left
end-marker $\ |\hspace{-2.0mm}C_{i}$ and terminated by the rihgt
end-marker $\$_{i}$;

\item $\delta$ is the transition function:
\begin{equation}
S\setminus(S_{acc}\cup S_{rej})\times\Gamma\rightarrow S\times
D_{1} \times D_{2}
\end{equation}
where $D_{1}=D_{2}=\{\leftarrow, \downarrow, \rightarrow\}$, of
which the element in turn means left, stationary and right,
respectively; we restrict that at least one of $D_{1}$ and $D_{2}$
is stationary in this mapping, which means that a 2TFA can at most
move one head and scan one new symbol at one time. If both of the
tape heads are stationary, the symbol to be scanned is the same as
the previous step.  We assume that the first symbol scanned is \
$|\hspace{-2.0mm}C_{1}$. So we can simplify the transition
function to take ``$\Gamma$" as variable instead of \
``$\Gamma\times\Gamma$", which is the main difference from the
original model introduced by Rabin and Scott. The mapping can be
interpreted as follows: $\delta(s, \gamma)=(s', d_{1}, d_{2})$
means a machine currently in state $s\in S\setminus(S_{acc}\cup
S_{rej})$, scanning symbol $\gamma\in \Gamma$ will change state to
$s'$, and the tape head $i$ will move according to the directions
$d_{1}$ and $d_{2}$;

\item $s_{0}\in S$ is the initial state of the machine;

\item $S_{acc}\subset S$ and $S_{rej}\subset S$ are the sets of
accepting states and rejecting states, respectively.
\end{itemize}
For an input string $\omega=\sigma_{1}\dots\sigma_{l}\in\Sigma^*$,
 the content of the first tape is $\mathcal{T}_1=\
|\hspace{-2.0mm}C_{1} \sigma_{1}\dots\sigma_{l}\$_1$, and the
content of the second tape is $\mathcal{T}_2=\
|\hspace{-2.0mm}C_{2} \sigma_{1}\dots\sigma_{l}\$_2$. The
computation of a 2TFA on input string $\omega$ is the sequence
$(p_0, i^1_0, i^2_0), \dots$, $(p_j,i^1_j,i^2_j), \dots$, in which
$p_j\in S$, $0\leq i^1_j,i^2_j\leq l+1$, and

\begin{itemize}
\item The initial configuration $(p_0,i^1_0,i^2_0)=(s_0,0,0)$;

\item If the instance configuration is $(p_j,i^1_j,i^2_j)$ and the
transition function is $\delta(p_{j},\sigma)=(p_{j+1},
d^{1}_{j},d^2_{j})$, then the next configuration will be
$(p_{j+1}, i^1_j+d^{1}_{j},i^2_j+d^2_{j} )$.

\end{itemize}
A computation is assumed to halt if and only if an accepting state
or a rejecting state is reached. If there is a configuration
$(p_j,i^1_{j},i^2_{j})$ with $p_j\in S_{acc}$, then the automaton
is said to accept $\omega$; else if $p_j\in S_{rej}$, then the
automaton is said to reject $\omega$.

For related basic notations and further details about {\it finite
automata}, the reader may refer to \cite{Hop,MS,SY}.

\subsection{Definition of 2TQCFA}

2QCFA were introduced by  Ambainis and  Watrous in \cite{AJ}. The
new model that we define here is similar to 2QCFA, but with two
tapes. Informally, we describe a 2TQCFA as a 2TFA which has access
to a constant size of quantum register, upon which it perform
quantum transformations and measurements. Before giving a more
formal definition of 2TQCFA, we review a few basic facts regarding
quantum computing. We would refer the reader to
\cite{Gru99,JP,Nie} for a more detailed overview of quantum
computing. Let $\mathcal{H}(Q)$ represents the Hilbert space with
the corresponding base identified with set $Q$. If
$\{|i_{Q}\rangle\}$ is an orthonormal base for $\mathcal{H}(Q)$,
then state $q\in \mathcal{H}(Q)$ can be written in the form
$\sum^{n}_{i}a_{i}|i_{Q}\rangle$, where $a_{i}$ is a complex
number and $\sum^{n}_{i}|a_{i}|^{2}=1$. Let
$\mathcal{U}(\mathcal{H}(Q))$ and $\mathcal{M}(\mathcal{H}(Q))$
denote the sets of
 unitary operators and orthogonal measurements over $\mathcal{H}(Q)$, respectively.  A unitary operator $U\in \mathcal{U}(\mathcal{H}(Q))$  is an invertible linear operator which preserves length.  A set of {\it projective measurements}
$\mathcal{M}=\{P_{i}\}\in \mathcal{M}(\mathcal{H}(Q))$ are {\it
Hermitian operators} on $\mathcal{H}(Q)$ such that
$P_{i}=P_{i}^{\dagger}$, $P_{i}^{2}=P_{i}^{\dagger}$ for all $i$;
$P_{i}P_{j}=0$ for $i\neq j$, and $\Sigma_{i}P_{i}=I$. Upon
measuring the state $|\psi\rangle$, the probability of getting
result $i$ is given by $p(i)=\langle\psi|P_{i}|\psi\rangle$. Given
that the result $i$ occurred, the state of the quantum system
immediately changes to
$P_{i}|\psi\rangle/\sqrt{\langle\psi|P_{i}|\psi\rangle}$.

A 2TQCFA is specified by a 9-tuple
\begin{equation}
\mathcal{M}=(Q,S,\Sigma,\Theta,\delta,q_{0},s_{0},S_{acc},S_{rej})
\end{equation}
where,
\begin{itemize}
\item $Q$ is a finite set of quantum states;

\item $S$ is a finite set of classical states;

\item $\Sigma$ is a finite set of input symbols; the tape symbol
set is $\Gamma=\Sigma\cup\{\ |\hspace{-2.0mm}C,\$\}$; two tapes
with the same input $\omega$ are
$\mathcal{T}_{1}=|\hspace{-2.0mm}C_{1}\omega\$_{1}$ and
$\mathcal{T}_{2}=|\hspace{-2.0mm}C_{2}\omega\$_{2}$; we assume
that the input in $\mathcal{T}_{i}$ is started by the left
end-marker $\ |\hspace{-2.0mm}C_{i}$ and terminated by the right
end-marker $\$_{i}$;

\item $\Theta$ is the transition function of quantum states:
\begin{equation}
S\setminus(S_{acc}\cup S_{rej})\times \Gamma\rightarrow
\mathcal{U}(\mathcal{H}(Q))\cup \mathcal{M}(\mathcal{H}(Q)),
\end{equation}
where each action $\Theta(s,\gamma)$ corresponds to either a
unitary transformation or a projective measurement;

\item $\delta$ is the transition function of classical states. If
$\Theta(s, \gamma)\in \mathcal{U}(\mathcal{H}(Q))$, then $\delta$
is:
\begin{equation}
S\setminus(S_{acc}\cup S_{rej})\times \Gamma\rightarrow S\times
D_{1}\times D_{2},
\end{equation}
which is the same transition function as we have defined in 2TFA.
If $\Theta(s, \gamma)\in \mathcal{M}(\mathcal{H}(Q))$, and the
resulting set of the measurement is $R=\{r_{1},r_{2},\dots$,
$r_{n}\}$, then
 $\delta$ is:
\begin{equation}
S\setminus(S_{acc}\cup S_{rej})\times \Gamma \times R\rightarrow
S\times D_{1}\times D_{2},
\end{equation}
 where $\delta(s,\gamma)(r_{i})=(s',d_{1},d_{2})$ means
that when the projective measurement result is $r_{i}$ with
probability $p_{i}$, then $s\in S$ scanning $\gamma\in \Gamma$
will be changed to state $s'$, and the next scanned symbol will be
decided by $d_{1}$ and $d_{2}$;

\item $q_{0}\in Q$ is the initial quantum state;

\item $s_{0}\in S$ is the initial classical state;

\item $S_{acc}\subset S$ and $S_{rej}\subset S$ are the sets of
classical accepting states and rejecting states, respectively.
\end{itemize}

Given an input $\omega$, a 2TQCFA
$\mathcal{M}=(Q,S,\Sigma,\Theta,\delta,q_{0},s_{0},S_{acc},S_{rej})$
operates as follows:

At the beginning, tape head $i$ is at \ $|\hspace{-2.0mm}C_{i}$,
the quantum initial state is $|q_{0}\rangle$, the classical
initial state is $s_{0}$, and $|q_{0}\rangle$ will be changed
according to $\Theta(s_{0},\ |\hspace{-2.0mm}C_{1})$.
\begin{itemize}

\item [a.] If $\Theta(s_{0},\ |\hspace{-2.0mm}C_{1})=U\in
\mathcal{U}(\mathcal{H}(Q))$, then the quantum state evolves as
$|q_{0}\rangle\rightarrow U|q_{0}\rangle$, and meanwhile, the
classical state $s_{0}$ will be changed to $s'$ according to
$\delta(s_{0},\ |\hspace{-2.0mm}C_{1})=(s',d_{1},d_{2})$. The
symbol scanned next is in $\mathcal{T}_{1}$, if
$d_1\neq\downarrow$; the symbol scanned next is in
$\mathcal{T}_{2}$, if $d_2\neq\downarrow$; otherwise, it holds
that $d_1=d_{2}=\downarrow$, and then the symbol scanned next is
unchanged.

\item [b.] If $\Theta(s_{0},\ |\hspace{-2.0mm}C_{1})=M\in
\mathcal{M}(\mathcal{H}(Q))$, $M=\{P_1, \dots, P_m\}$, then
$|q_{0}\rangle\rightarrow \frac{P_{i}|q_{0}\rangle}{\sqrt{\langle
q_{0}|P_{i}|q_{0}\rangle}}$ with result $r_i$. Meanwhile, we have
$\delta(s_{0},\
|\hspace{-2.0mm}C_{1})(r_{i})=(s_{i},d_{1},d_{2})$. If $s_{i}\in
S_{acc}$, $\mathcal{M}$ accepts $\omega$; if $s_{i}\in S_{rej}$,
$\mathcal{M}$ rejects $\omega$; otherwise, $\mathcal{M}$ scans the
next symbol, and the transformation is similar to the above
process.
\end{itemize}
A computation is assumed to halt if and only if an accepting state
or a rejecting classical state is reached.  Let $p_{acc}(\omega)$
and $p_{rej}(\omega)$ denote the accepting and rejecting
probabilities for $\mathcal{M}$ with input $\omega$, respectively.
 We say that a language
$L$ is recognized by a 2TQCFA $\mathcal{M}$ if $L$ is recognized
by a 2TQCFA $\mathcal{M}$ with one-sided error $\epsilon$, which
means that for all $\omega\in \Sigma^{*}$,
$p_{acc}(\omega)+p_{rej}(\omega)=1$, $p_{acc}(\omega)=1$ if
$\omega\in L$, and $p_{rej}(\omega)\geq 1-\epsilon$ if $\omega
\notin L$.

\subsection{Definitions of $m$TFA and $m$TQCFA}

The definition of $k$TFA is similar to that of 2TFA. A $k$TFA is a
6-tuple
\begin{equation}
M=(S,\Sigma,\delta,s_{0},S_{acc},S_{rej})
\end{equation}
where $S,\Sigma,s_{0},S_{acc}$ and $S_{rej}$ are just like the
ones in 2TFA. The map of $\delta$ is: $S\setminus(S_{acc}\cup
S_{rej})\times\Gamma\rightarrow S\times D^{k}$. We restrict that
at most one of $D_{i}(i=1,2,\dots,k)$ is no stationary in the
mapping and all the tapes have the same input $\omega$. The
transformation is similar to that of  2TFA.

An $m$TQCFA  is essentially a classical $m$TFA augmented with a
quantum component of constant size.

\section{Languages recognized by 2TFA}

When we use a 2DFA to recognize a language, we can not record how
many specific symbols have been scanned. For example, when we use
a 2DFA to recognize $L_{eq}=\{a^{n}b^{n}\mid n\in \mathbf{N}\}$,
we can not record how many symbols $a$ have been scanned, which is
the reason that 2DFA can not recognize $L_{eq}$. Even if a 2TFA is
used, we are still not able to record how many specific symbols
have been scanned.
 But we have two tapes, so we can compare the numbers of the
specific symbols having been scanned between the two tapes, which
leads to that 2TFA can recognize more languages and become more
powerful.

\subsection{Languages recognized by 2QFA or 2QCFA that can also be recognized by 2TFA}
$L_{eq}$ and $L_{trieq}$ were proved to be recognized by 2QFA in
\cite{Kon97},  $L_{pal}$ and $L_{eq}$ were proved to be recognized
by 2QCFA in \cite{AJ}, and $M_{eq}$, $L_{eq}(k,a)$ and $L_{=}$
were proved to be recognized by 2QCFA in \cite{Qiu2}. In this
subsection we will prove all these languages can be recognized by
2TFA in linear time. We assume that
$\mathbf{N}=\{0,1,2,\dots,n,\dots\}$, and we use
$\#_{\omega}(\sigma)$ and $|\omega|$  to denote the number of
symbol $\sigma$ in string $\omega$ and the length of string
$\omega$, respectively.
\begin{Th}
$L_{eq}=\{a^{n}b^{n}\mid n\in \mathbf{N}\}$ can be recognized by
2TFA in linear time.
\end{Th}
\begin{proof}The idea of this proof is as follows: First we check whether
the input $\omega$ is of the form $a^{+}b^{+}$, and then we scan
$\mathcal{T}_1$ from right to left to consume symbol $b$ and scan
$\mathcal{T}_2$ from left to right to consume symbol $a$
alternately. If we scan a symbol $b$ in $\mathcal{T}_1$, the next
symbol scanned in $\mathcal{T}_2$ must be symbol $a$; otherwise,
we reject the input $\omega$. Once we scan a symbol $a$ in
$\mathcal{T}_1$, which means that the tape head of $\mathcal{T}_1$
reaches the $a^+$ section of $\mathcal{T}_1$, we turn to
$\mathcal{T}_2$ to see which symbol it scans. If the next symbol
scanned in $\mathcal{T}_2$ is symbol $b$, we accept the input
$\omega$; otherwise, we reject it. A 2TFA $\mathcal{M}$ for
$L_{eq}$ is defined as follows:
\begin{equation}
\mathcal{M}=(S,\Sigma,\delta,s_{0},S_{acc},S_{rej})
\end{equation}
where,
\begin{itemize}

\item $S=\{s_0,s_1,s_2,s_a,s_b,s_{crt},s_{acc},s_{rej}\}$, where
$s_0,s_1$ and $s_2$ are used to check whether $\omega$ is of the
form $a^{+}b^{+}$; $s_a$ is used to consume symbol $a$ in
$\mathcal{T}_2$ and $s_b$ is used to consume symbol $b$ in
$\mathcal{T}_1$; $s_{crt}$ is the critical state which  will be
changed to $s_{acc}$ or $s_{rej}$ depending on the next scanned
symbol; \item $\Sigma=\{a,b\}$; $s_{0}$ is the initial state;
$S_{acc}=\{s_{acc}\}$; $S_{rej}=\{s_{rej}\}$;

 \item The mapping $\delta$ is defined as follows:\\

\begin{tabular}{lll}
  $\delta(s_0,\ |\hspace{-2.0mm}C_1)=(s_0, \rightarrow,\downarrow)$
  &$\delta(s_0, a)=(s_1, \rightarrow, \downarrow)$&{}\\
  $\delta(s_0, b)=(s_{rej}, -, -)$
  &$\delta(s_0, \$_1)=(s_{acc}, -, -)$&{}\\
  $\delta(s_1, a)=(s_1, \rightarrow, \downarrow)$
  &$\delta(s_1, b)=(s_2, \rightarrow, \downarrow)$
  &$\delta(s_1, \$_1)=(s_{rej}, -, -)$\\
  $\delta(s_2, a)=(s_{rej}, -, -)$
  &$\delta(s_2, b)=(s_2, \rightarrow, \downarrow)$\
  &$\delta(s_2, \$_1)=(s_{b}, \leftarrow, \downarrow)$\\
  $\delta(s_b, a)=(s_{crt}, \downarrow, \rightarrow)$
  &$\delta(s_b, b)=(s_a, \downarrow, \rightarrow)$&{} \\
  $\delta(s_a, a)=(s_b, \leftarrow, \downarrow)$
  &$\delta(s_a, b)=(s_{rej}, -,-)$&{}\\
  $\delta(s_{crt}, a)=(s_{rej}, -, -)$\
  &$\delta(s_{crt}, b)=(s_{acc}, -, -)$&{}\\
\end{tabular}
\\
\\
{\it Remark :}When $\mathcal{M}$ comes to halting states, the
directions of the tape heads will be meaningless, and thus we use
``$-$".
\end{itemize}

It is easy to verify that $L_{eq}$ can be recognized by
$\mathcal{M}$. Checking whether the input $\omega$ is of the form
$a^{+}b^{+}$ takes  $|\omega|$ time. Scanning $\mathcal{T}_1$ from
right to left and scanning $\mathcal{T}_2$ from left to right
alternately takes $2|\omega|$ time at the worst case. Hence, the
worst running time of $\mathcal{M}$ is $\mathbf{O}(|\omega|)$,
which is linear time.

\end{proof}

\begin{Th}
$L_{pal}=\{\omega\in\{a,b\}^{*}\mid\omega=\omega^{R}\}$ can be
recognized by 2TFA in linear time.
\end{Th}
\begin{proof}We scan $\mathcal{T}_1$ from right to left and
scan $\mathcal{T}_2$ from left to right alternately, and each time
we make sure that the symbols scanned are the same. If the symbols
scanned are not the same, we reject the input $\omega$; otherwise,
when the tape head of $\mathcal{T}_1$ reaches the left end-marker
and the tape head of $\mathcal{T}_2$ reaches the right end-marker,
we accept it. A 2TFA $\mathcal{M}$ for $L_{pal}$ is defined as
follows:
\begin{equation}
\mathcal{M}=(S,\Sigma,\delta,s_{0},S_{acc},S_{rej})
\end{equation}
where,
\begin{itemize}

\item $S=\{s_0,s_t,s,s_a,s_b,s_{crt},s_{acc},s_{rej}\}$, where
$s_0$ and $s_t$ are used to move the first tape head to the right
end-marker of $\mathcal{T}_1$; $s_a$ and $s_b$ are used to consume
symbol $a$ and $b$ in $\mathcal{T}_2$, respectively; $s$ is used
to consume a symbol in $\mathcal{T}_1$; $s_{crt}$ is the critical
state which will be changed to $s_{acc}$ or $s_{rej}$ depending on
the next scanned symbol;

\item $\Sigma=\{a,b\}$; $s_{0}$ is the initial state;
$S_{acc}=\{s_{acc}\}$; $S_{rej}=\{s_{rej}\}$.

 \item The mapping $\delta$ is defined as follows:\\

 \begin{tabular}{lll}
 $\delta(s_0, \sigma)=(s_0, \rightarrow, \downarrow)$ for
 $\sigma=a,b,\ |\hspace{-2.0mm}C_1$&\ \ &
 $\delta(s_0, \$_1)=(s_t, \leftarrow, \downarrow)$\\
 $\delta(s_t, \sigma)=(s_{\sigma}, \downarrow, \rightarrow)$ for
 $\sigma=a,b$&\ \ \ &
 $\delta(s_t,\ |\hspace{-2.0mm}C_1)=(s_{acc}, -, -)$\\
 $\delta(s_{\sigma}, \sigma')=(s, \leftarrow, \downarrow)$ for
 $\sigma=\sigma'$&\ \ &
 $\delta(s_{\sigma}, \sigma')=(s_{rej},-,-)$ for
 $\sigma\neq\sigma'$\\
 $\delta(s, \sigma)=(s_{\sigma}, \downarrow, \rightarrow)$ for
 $\sigma=a,b$&\ \ \ &
 $\delta(s,\ |\hspace{-2.0mm}C_1)=(s_{crt}, \downarrow, \rightarrow)$\\
 $\delta(s_{crt}, \sigma)=(s_{rej},-, -)$ for $\sigma\neq \$_2$&\ \ \ &
 $\delta(s_{crt}, \$_2)=(s_{acc},-, -)$\\
 \end{tabular}

\end{itemize}
It is easy to verify that $L_{pal}$ can be recognized by
$\mathcal{M}$. It takes $\mathcal{M}$ exactly $3|\omega|$ time to
accept the input $\omega$ and less than $3|\omega|$ time to reject
it. So $L_{pal}$ can be recognized by 2TFA in linear time.
 \end{proof}

\begin{Th}
$L_{trieq}=\{a^{n}b^{n}c^{n}\mid n\in \mathbf{N}\}$ can be
recognized by 2TFA in linear time.
\end{Th}
\begin{proof}
The 2TFA $\mathcal{M}$ to recognize this language is similar to
the one recognizing $L_{eq}$. The idea of this proof is as
follows: First, we check whether the input $\omega$ is of the form
$a^{+}b^{+}c^{+}$; second, we use two subroutines of 2TFA
$\mathcal{M}'$  which recognizes $L_{eq}$ to check whether the
number of symbol $c$ is equal to the number of symbol $a$ and $b$.
We scan $\mathcal{T}_1$ from right to left and $\mathcal{T}_2$
from left to right alternately to make sure the number of symbol
$c$ in $\mathcal{T}_1$ is equal to the number of symbol $a$ in
$\mathcal{T}_2$. We scan $\mathcal{T}_1$ from left to right and
$\mathcal{T}_2$ from left to right alternately to make sure the
number of symbol $c$ in $\mathcal{T}_1$ is equal to the number of
symbol $b$ in $\mathcal{T}_2$. The details of 2TFA $\mathcal{M}$
for $L_{trieq}$ are omitted here.

It is easy to verify that $L_{trieq}$ can be recognized by
$\mathcal{M}$. It takes $\mathcal{M}$ $3|\omega|$ time to accept
the ipnut $\omega$ and less than $3|\omega|$ to reject it at the
worst cases. So $L_{trieq}$ can be recognized by 2TFA in linear
time.
\end{proof}

\begin{Th}\label{th4}
$L_{eq}(k,a)=\{a^{kn}b^{n}\mid n\in \mathbf{N}\}$ can be
recognized by 2TFA in linear time.
\end{Th}
\begin{proof}
The 2TFA $\mathcal{M}$ to recognize this language is similar to
the one recognizing
 $L_{eq}$. First we check whether the input $\omega$ is of the form
$a^{+}b^{+}$. But, we scan $\mathcal{T}_1$ from right to left to
consume one symbol $b$ and scan $\mathcal{T}_2$ from left to right
to consume $k$ symbols $a$ instead of just one symbol alternately.
A 2TFA $\mathcal{M}$ for $L_{eq}(k,a)$ is defined as follows:
\begin{equation}
\mathcal{M}=(S,\Sigma,\delta,s_{0},S_{acc},S_{rej})
\end{equation}
where,
\begin{itemize}
\item
$S=\{s_0,s_1,s_2,s_{a1},s_{a2},\dots,s_{ak},s_{b},s_{crt},s_{acc},s_{rej}\}$,
where $s_0,s_1$ and $s_2$ are used to check whether $\omega$ is of
the form $a^{+}b^{+}$; $s_{ai}$ is used to consume symbol $a$ in
$\mathcal{T}_2$ and $s_{b}$ is used to consume symbol $b$ in
$\mathcal{T}_1$; $s_{crt}$ is the critical state which  will be
changed to $s_{acc}$ or $s_{rej}$ depending on the next scanned
symbol;

\item $\Sigma=\{a,b\}$; $s_{0}$ is the initial state;
$S_{acc}=\{s_{acc}\}$; $S_{rej}=\{s_{rej}\}$;

\item The mapping $\delta$ is defined as follows:\\

\begin{tabular}{lll}
  $\delta(s_0,\ |\hspace{-2.0mm}C_1)=(s_0,
  \rightarrow,\downarrow)$\ \ \ \ &
  $\delta(s_0, a)=(s_1, \rightarrow, \downarrow)$&{ }\\
  $\delta(s_0, b)=(s_{rej}, -, -)$&
  $\delta(s_0, \$_1)=(s_{acc}, -, -)$&{ }\\
  $\delta(s_1, a)=(s_1, \rightarrow, \downarrow)$&
  $\delta(s_1, b)=(s_2, \rightarrow, \downarrow)$&
  $\delta(s_1, \$_1)=(s_{rej}, -, -)$\\
  $\delta(s_2, a)=(s_{rej}, -, -)$&
  $\delta(s_2, b)=(s_2, \rightarrow, \downarrow)$&
  $\delta(s_2, \$_1)=(s_{b}, \leftarrow, \downarrow)$\\
  $\delta(s_b, a)=(s_{crt}, \downarrow, \rightarrow, )$&
  $\delta(s_b, b)=(s_{a1}, \downarrow, \rightarrow, )$&{ }\\
  $\delta(s_{ai}, a)=(s_{a(i+1)}, \downarrow,\rightarrow)$ & for $i=1,2,\dots,k-1$&{ }\\
  $\delta(s_{ai}, b)=(s_{rej}, -,-)$ & for $i=1,2,\dots,k-1$&{ }\\
  $\delta(s_{ak}, a)=(s_b, \leftarrow, \downarrow)$&
  $\delta(s_{ak}, b)=(s_{rej}, -,-)$&{ }\\
  $\delta(s_{crt}, a)=(s_{rej}, -, -)$&
  $\delta(s_{crt}, b)=(s_{acc}, -, -)$&{ }\\
  \end{tabular}

\end{itemize}
It takes $\mathcal{M}$ $3|\omega|$ time at the worst cases. So
$L_{eq}(k,a)$ can be recognized by 2TFA in linear time.
\end{proof}

\begin{Th}
$M_{eq}=\{a^{n}b_{1}^{n}a^{m}b_{2}^{m}\mid n,m\in \mathbf{N}\}$
can be recognized by 2TFA in linear time.
\end{Th}

\begin{proof}
 The 2TFA $\mathcal{M}$ to recognize this language is similar to the one recognizing
 $L_{trieq}$. We check whether the input $\omega$ is of the form
$a^{+}b_{1}^{+}a^{+}b_{2}^{+}$ at the beginning. Next we use two
subroutines of 2TFA $\mathcal{M}'$ which recognizes $L_{eq}$.
First,we scan $\mathcal{T}_1$ from left to right and
$\mathcal{T}_2$ from left to right alternately to make sure the
number of symbol $a$ in the first section of $\mathcal{T}_1$ is
equal to the number of symbol $b_1$ in $\mathcal{T}_2$. Second, we
scan $\mathcal{T}_1$ from left to right and $\mathcal{T}_2$ from
left to right alternately to make sure the number of symbol $a$ in
the second section of $\mathcal{T}_1$ is equal to the number of
symbol $b_2$ in $\mathcal{T}_2$. We do not bore the details of
 2TFA $\mathcal{M}$ for $M_{eq}$ here.

It is easy to verify that $M_{eq}$ can be recognized by
$\mathcal{M}$. It takes $\mathcal{M}$ $2|\omega|$ time at the
worst cases. So $M_{eq}$ can be recognized by 2TFA in linear time.
\end{proof}

\begin{Th}
$L_{=}=\{\omega\in \{a,b\}^{*}\mid\#_{\omega}(a)=\#_{\omega}(b)\}$
can be recognized by 2TFA in linear time.
\end{Th}

\begin{proof}
 The 2TFA $\mathcal{M}$ to recognize this language is
similar to the one recognizing $L_{eq}$. But we do not need to
check whether $\omega$ is of the form $a^{+}b^{+}$. We scan
$\mathcal{T}_1$ for symbol $a$ from left to right and scan
$\mathcal{T}_2$ for symbol $b$ from left to right alternately. We
start at looking for symbol $a$ in $\mathcal{T}_1$. If the symbol
scanned is not symbol $a$, we skip it and look for the next symbol
in $\mathcal{T}_1$ until we scan symbol $a$. Once we have scanned
a symbol $a$ in $\mathcal{T}_1$, we look for a symbol $b$ in
$\mathcal{T}_2$. If the symbol scanned is not symbol $b$, we just
skip it, and look for the next symbol in $\mathcal{T}_2$ until we
scan symbol $b$ or come to the right end-marker. If the head of
$\mathcal{T}_2$ reaches the right end-marker, we reject the input
$\omega$; otherwise, turn to $\mathcal{T}_1$. When the tape head
of $\mathcal{T}_1$ reaches the right end-marker, if there is no
more symbol $b$ waiting for being scanned in $\mathcal{T}_2$, we
accept the input $\omega$; otherwise, we reject it. A 2TFA
$\mathcal{M}$ for $L_{eq}$ is defined as follows:
\begin{equation}
M=(S,\Sigma,\delta,s_{0},S_{acc},S_{rej})
\end{equation}
where,
\begin{itemize}

\item $S=\{s_0,s_a,s_b,s_{crt},s_{g},s_{l},s_{eq}\}$, where $s_a$
is used to consume symbol $a$ in $\mathcal{T}_1$ and $s_b$ is used
to consume symbol $b$ in $\mathcal{T}_2$; $s_{crt}$ is the
critical state to check whether $\omega$ is in $L$; $s_{g}$ means
$\#_{\omega}(a)>\#_{\omega}(b)$, $s_{l}$ means
$\#_{\omega}(a)<\#_{\omega}(b)$, and $s_{eq}$ means
$\#_{\omega}(a)=\#_{\omega}(b)$, respectively;

\item $\Sigma=\{a,b\}$; $s_{0}$ is the initial state;
$S_{acc}=\{s_{eq}\}$; $S_{rej}=\{s_g, s_l\}$;

 \item The mapping $\delta$ is defined as follows:\\

 \begin{tabular}{lll}
  $\delta(s_0, |\hspace{-2.0mm}C_1)=(s_0,
  \rightarrow,\downarrow)$&
  $\delta(s_0, \$_1)=(s_{acc}, -, -)$& \\

  $\delta(s_0, a)=(s_b, \downarrow, \rightarrow, )$&
  $\delta(s_0, b)=(s_0, \downarrow, \rightarrow, )$& \\

  $\delta(s_b, a)=(s_{b}, \downarrow, \rightarrow)$&
  $\delta(s_b, b)=(s_a, \rightarrow, \downarrow)$&
  $\delta(s_b, \$_2)=(s_g,-, -)$\\
  $\delta(s_a, a)=(s_b, \downarrow, \rightarrow)$&
  $\delta(s_a, b)=(s_a, \rightarrow, \downarrow)$&
  $\delta(s_a,\$_1)=(s_{crt}, \downarrow, \rightarrow)$\\
  $\delta(s_{crt},a)=(s_{crt}, \downarrow, \rightarrow)$\ &
  $\delta(s_{crt},b)=(s_{l}, -,-)$\ &
  $\delta(s_{crt},\$_2)=(s_{eq}, -, -)$
\end{tabular}

\end{itemize}

It is easy to verify that $L_{=}$ can be recognized by
$\mathcal{M}$. It takes $\mathcal{M}$ $2|\omega|$ time to accept
the input $\omega$ and no more than $2|\omega|$ time to reject it.
So $L_{=}$ can be recognized by 2TFA in linear time.
\end{proof}

Having proved several languages that can be recognized by 2TFA,
next we will prove a more general language
$L_{linear}=\{\sigma_{1}^{k_1}\sigma_{2}^{k_2}\dots
\sigma_{m}^{k_m} \mid a_{i}k_{i}=\sum_{j=1,j\neq i}^{m}a_{j}k_{j},
\Sigma=\{\sigma_{1},\sigma_{2},\dots, \sigma_{m}\},\linebreak[0]
m,{k_1},{k_2},\dots {k_m}, a_i\in \mathbf{N};\
a_1,a_2,\dots,a_{i-1},\linebreak[1] a_{i+1},\dots, a_m\in
\{0,1\}\}$ can be recognized by 2TFA in linear time. $L_{eq}$,
$L_{trieq}$, $M_{eq}$ and $L_{eq}(k,a)$ are just the special cases
of this language. $L_{linear}$  can also be recognized by 2QCFA in
polynomial time.

\begin{Th}
$L_{linear}=\{\sigma_{1}^{k_1}\sigma_{2}^{k_2}\dots
\sigma_{m}^{k_m} \mid a_{i}k_{i}=\sum_{j=1,j\neq
i}^{m}a_{j}k_{j},\ \Sigma=\{\sigma_{1},\dots,\linebreak[0]
\sigma_{m}\}, m,{k_1},\linebreak[0] {k_2},\dots {k_m}, a_i\in
\mathbf{N};\ a_1,a_2,\dots,a_{i-1},a_{i+1},\dots, a_m\in
\{0,1\}\}$ can be recognized by 2TFA in linear time.
\end{Th}

\begin{proof}
 $\omega\in L_{linear}$ has at most $m-1$
restrictions such as $a_{i}k_{i}=\sum_{j=1,j\neq i}^{m}a_{j}k_{j}$
(if it has more than $m-1$ restrictions, they must be relative,
and we can replace them by less than $m-1$ restrictions), where
$m=|\Sigma|$ which is a constant. First, we use $\mathcal{T}_1$ to
check whether $\omega$ is of the form
$\sigma_{1}^{*}\sigma_{2}^{*}\dots \sigma_{m}^{*}$, and then we
move the tape head of $\mathcal{T}_1$ back to the left end-maker.
This takes $2|\omega|$ time. The 2TFA $\mathcal{M}$ for
$L_{linear}$ consists of by $m-1$ subroutines, and each subroutine
verifies one restriction. 2TFA $\mathcal{M}$ accepts $\omega$ if
all restrictions are satisfied. As an example, for a restriction
$a_{i}k_{i}=\sum_{j=1,j\neq i}^{m}a_{j}k_{j}$, symbol $\sigma_i$
can be regarded as symbol $a$, $\sigma_j(j\neq i, a_j\neq 0)$ can
be regarded as symbol $b$, and $\sigma_j(j\neq i, a_j= 0)$ can be
regarded as symbol $c$. This subroutine is equivalent to the
subroutine of recognizing $\{\omega\in\{a,b,c\}^{*} \mid
a_i\#_{\omega}(a)=\#_{\omega}(b)\}$, which is easy to implement.
We scan $\mathcal{T}_1$ from left to right to consume one symbol
$a$ and scan $\mathcal{T}_2$ from left to right to consume $a_i$
symbols $b$, and we just skip symbol $c$. This subroutine takes
$2|\omega|$ time. After a subroutine ends, we move both tape heads
to the left end-markers, which takes $2|\omega|$ time. So
$L_{linear}$ can be recognized by a 2TFA with running time
$(4(m-1)+2)|\omega|$ at the worst cases, which is linear time.

\end{proof}

{\it Remark :} $L_{linear}$ can be recognized by 2QCFA. We can
used m-1 subroutines to verify the m-1 restrictions, and the idea
is similar to the proof of the above theorem.

\subsection{Languages recognized by 2TFA but whether or not they can be recognized by 2QFA or 2QCFA is still pending}
In this subsection we will show that several languages can be
recognized by 2TFA in linear time, including  $L_{copy}$,
$L_{middle}$ and $L_{balanced}$. Whether $L_{middle}$ and
$L_{balanced}$ can be recognized by 2QCFA or 2QFA or not still
remains as an open problem \cite{AJ}.

\begin{Th}

$L_{copy}=\{\omega\omega\mid\omega\in\Sigma^{*}\}$  can be
recognized by 2TFA in linear time.
\end{Th}

\begin{proof}
First, we check whether the length of the input is even. If it is
odd, the input is rejected; otherwise, we continue to do the next
steps. The key to solve this problem is to find the middle of the
input. In terms of the advantage of two tapes, we can easily find
the middle as follows: Tape $\mathcal{T}_1$ scans two symbols from
left to right while tape $\mathcal{T}_2$ scans one symbol from
left to right. When tape $\mathcal{T}_1$ reaches the right
end-marker, tape $\mathcal{T}_2$ will just go to the middle
exactly. After doing this, we move the tape head of
$\mathcal{T}_1$ to the left end-maker \ $|\hspace{-2.0mm}C_1$, and
then compare the symbols between $\mathcal{T}_1$ and
$\mathcal{T}_2$ one by one from left to right until tape
$\mathcal{T}_2$ reaches the end. If all the symbols compared are
the same, the input is accepted; otherwise, it is rejected.
Obviously, this can be done in linear time.
\end{proof}

\begin{Th}
$L_{middle}=\{xay\mid x,y\in\Sigma^{*},a\in\Sigma\}$  can be
recognized by 2TFA in linear time.
\end{Th}

\begin{proof}
First, we check whether the length of the input is odd. If it is
even, the input is rejected; otherwise, we continue to do the
following steps. The key to solve this problem is still to find
the middle of the input. We can find it easily in the same way as
we did in the previous theorem. We check whether the symbol in the
middle is $a$. If the symbol is $a$, the input is accepted;
otherwise, it is rejected. Obviously, this can be done in linear
time.
\end{proof}

\begin{Th}
$L_{balanced}=\{x\in\{``(",``)"\}^*\mid$ parentheses in $x$ are
balanced $\}$ can be recognized by 2TFA in linear time.
\end{Th}

\begin{proof}
 In order to prove this theorem, we need to prove a lemma. We use $F(w)$  to denote the set of all the prefixes of $w$
except empty string $\lambda$ and $w$.

\begin{Lm}
A string $w$ is in $L_{balanced}$ if and only if $w$
satisfies\\ (1).$\#_w(``(")=\#_w(``)")$ and\\
(2).$\#_x(``(")\geq\#_x(``)")$ for all $x\in F(w)$.
\end{Lm}
\begin{proof}
\begin{enumerate}
 \item[{\it 1}]: First, we show that if $w$ satisfies $\#_w(``(")=\#_w(``)")$ and $\#_x(``(")$ $\geq\#_x(``)")$ for all $x\in
 F(w)$, then $w\in L_{balanced}$. We prove this by induction on
 the length of $w$.
 \begin{description}
   \item BASIS: $|w|=0$ is obvious. When $|w|=2$, let
   $w=ab$, where $a,b\in\{``(",``)"\}$. We have $F(w)=\{a\}$ according to the definition. Because
   $\#_a(``(")\geq\#_a(``)")$, $a$ must be $``("$. Because $\#_w(``(")=\#_w(``)")$,
   $b$ must be $``)"$. Therefore $w=``()"$, which follows that $w\in  L_{balanced}$.
   \item INDUCTION: Suppose that when $|w|\leq 2n (n\geq 1)$, the result holds. We will prove it holds for $|w|=2(n+1)$.
\begin{description}
  \item[Case1] There exists a string $u\in F(w)$ which satisfies
  $\#_u(``(")=\#_u(``)")$. In this case, let $w=uv(u,v\in\{``(",``)"\}^*)$.
  Obviously, $u$ satisfies condition(1) and (2), and meanwhile we have $|u|\leq 2n$. According to the assumption, we get $u\in L_{balanced}$.
  Furthermore, we have $\#_v(``(")=\#_w(``(")-\#_u(``(")=\#_w(``)")-\#_u(``)")=\#_v(``)")$. So we get $v$
  satisfies condition(1). For every $x\in F(v)$, we have
  $\#_x(``(")=\#_{ux}(``(")-\#_u(``(")\geq\#_{ux}(``)")-\#_u(``)")=\#_x(``)")$. So we get $v$
  satisfies condition (2). Because $|v|\leq 2n$, we get $v\in L_{balanced}$ according to the assumption. It follows that $w\in
  L_{balanced}$.
  \item[Case2] Otherwise, for every $u\in F(w)$ satisfies $\#_u(``(")>\#_u(``)")$. Let
  $w=aw'b$. According to condition(2), we have $\#_a(``(")\geq\#_a(``)")$,
  and thus we get $a=``("$. Because $|aw'|$ is
  odd, we have  $\#_{aw'}(``(")>\#_{aw'}(``)")$ according to condition (2). According to
  condition(1), we have $\#_w(``(")=\#_w(``)")$, and thus it holds that $b=``)"$. It follows that
  $\#_{w'}(``(")=\#_w(``(")-1=\#_w(``)")-1=\#_{w'}(``)")$. Thus $w'$
  satisfies condition(1). For every $v\in F(w')$,
  we have $\#_v(``(")=\#_{av}(``(")-1>\#_{av}(``)")-1\geq\#_{av}(``)")=\#_{v}(``)")$,
  thus $w'$ satisfies condition(2). Meanwhile, we have $|w'|=2n$, and thus we get $w'\in L_{balanced}$ according to
the assumption. It follows that $w\in L_{balanced}$.
\end{description}
 \end{description}

 \item[{\it 2}]: We prove the other direction here. If $w\in L_{balanced}$,
 then $w$ satisfies condition (1) and (2). We prove this by induction
 on the length of $w$.

 \begin{description}
   \item BASIS: $|w|=0$ is obvious. When $|w|=2$, it must be the case
   $w``()"$. Obviously, it satisfies condition (1) and (2).
 \item INDUCTION: Suppose that when $|w|\leq 2n (n\geq 1)$, the result holds. We will prove it holds for
  $|w|=2(n+1)$.

   \begin{description}
  \item[Case1] $w$ can be divided into two parts as $uv$, where $u,v\in L_{balanced}$.
  Obviously, we have $|u|\leq 2n$ and $|v|\leq2n$. According to the assumption, we get $u$ and $v$
  satisfy condition (1) and (2). It follows that
  $\#_w(``(")=\#_u(``(")+\#_v(``(")=\#_u(``)")+\#_v(``)")=\#_w(``)")$, which says $w$ satisfies
  condition (1). For every $x\in F(w)$, if $|x|<=|u|$, then we have $x\in F(u)$ and $\#_x(``(")\geq \#_x(``)")$; otherwise, let $x=uy$, and then we have $y\in
  F(v)$. Therefore, it follows that $\#_x(``(")=\#_u(``(")+\#_y(``(")\geq \#_u(``)")+\#_y(``)")=\#_x(``)")$. So
  we get $w$ satisfies condition (2).

 \item[Case2] $w$ can not be divided into two parts as $uv$ such that $u,v\in L_{balanced}$.
  In this case, $w$ must be of the form $``("w'``)"$, where $w'\in L_{balanced}$
and $|w'|=2n$. According to the assumption, we have $w'\in
L_{balanced}$
 satisfies condition (1) and (2). It follows that
 $\#_w(``(")=\#_{w'}(``(")+1=\#_{w'}(``)")+1=\#_w(``)")$, which says $w$ satisfies
 condition(1). For every $x\in F(w)$, let $x=``("y$. Then we have $y\in
 F(w')$. Therefore we have $\#_x(``(")=1+\#_y(``(")\geq
1+\#_y(``)")>\#_y(``)")=\#_x(``)")$, i.e.,
 $w$ satisfies condition(2).
\end{description}

 \end{description}

\end{enumerate}
Therefore, the lemma has been proved.
\end{proof}
Now we return to the proof of the theorem. According to the lemma
proved above, it is enough for us to check whether the input $w$
satisfies condition (1) and (2). We use tape $\mathcal{T}_1$ to
scan the input $w$ from left to right, and use tape
$\mathcal{T}_2$ to save the result $R=\#_x(``(")-\#_x(``)")$ for
every $x\in F(w)$. If the tape head of $\mathcal{T}_2$ is on \
$|\hspace{-2.0mm}C_2$, it means that $R=0$; if the tape head of
$\mathcal{T}_2$ is on the $i$th square of $w$, it means that
$R=i$. When  tape $\mathcal{T}_1$ scans a symbol $``("$, the tape
head of $\mathcal{T}_2$ moves right by one square; when tape
$\mathcal{T}_1$ scans a symbol $``)"$, the tape head of
$\mathcal{T}_2$ moves left by one square. If the tape head of
$\mathcal{T}_2$ is on \ $|\hspace{-2.0mm}C_2$, and the symbol
scanned by the tape head of $\mathcal{T}_1$ is $``)"$, then the
input is rejected. When the tape head of $\mathcal{T}_1$ comes to
the right end-marker of $w$, and the tape head of $\mathcal{T}_2$
is on \ $|\hspace{-2.0mm}C_2$, then the input is accepted.
Obviously, this can be done in linear time. Therefore, we have
proved the theorem.
\end{proof}

\section{Languages recognized by 2TQCFA}
 Ambainis and  Watrous introduced a model of quantum
automata called 2QCFA and proposed an open problem of whether
$L_{square}=\{a^{n}b^{n^{2}}\mid n\in \mathbf{N}\}$ can be
recognized by 2QCFA or not \cite{AJ}. Although this problem
remains  open
 for 2QCFA, we will propose a new model of quantum automata
called 2TQCFA by which $L_{square}$ can be recognized in
polynomial time. In addition, Ambainis and  Watrous did not give
the details on how to use a 2QCFA to simulate a coin flip that is
an essential component for 2QCFA to recognize languages. Thus we
will give the details here.
\begin{Th}\label{th1}
For any $\epsilon>0$, there is a 2TQCFA $\mathcal{M}$ that accepts
any $\omega \in L_{square}=\{a^{n}b^{n^{2}}\mid n\in \mathbf{N}\}$
with certainly, rejects any $\omega \notin
L_{square}=\{a^{n}b^{n^{2}}\mid n\in \mathbf{N}\}$ with
probability at least $1-\epsilon$ and halts in polynomial time.
\end{Th}

\begin{proof}
The main idea is as follows: First, we  verify whether it holds
that $\omega\in\{a^{n}b^{kn}\mid n\in \mathbf{N}\}$. This can be
done in linear time, and the proof is similar to the proof of
Theorem~\ref{th4}. Second, we verify whether it holds that $k=n$.
We consider a 2TQCFA $\mathcal{M}$ with 2 quantum states
$|q_0\rangle$ and $|q_1\rangle$. $\mathcal{M}$ starts with the
quantum state $|q_0\rangle$. Firstly, we scan $\mathcal{T}_1$, and
every time when $\mathcal{M}$ scans symbol $a$ in $\mathcal{T}_1$,
the quantum state rotated by angle $\alpha=\sqrt 2\pi$. After all
symbols $a$ in $\mathcal{T}_1$ have been scanned, we scan  symbol
$b$ in $\mathcal{T}_1$ and  symbol $a$ in $\mathcal{T}_2$
alternately. When the number of symbol $b$ in $\mathcal{T}_1$
having been scanned equals to the total number of symbol $a$ in
$\mathcal{T}_2$, the quantum state is rotated by $-\alpha$. When
the end of $\mathcal{T}_1$ is reached, we measure the quantum
state. If the measurement result is $|q_1\rangle$, $\omega$ is
rejected. Otherwise, the process is repeated.

After having verified $\omega\in\{a^{n}b^{kn}\mid n\in
\mathbf{N}\}$, the quantum state is rotated by $\alpha$ when
$\mathcal{M}$ scans symbol $a$ in $\mathcal{T}_1$, and the quantum
state is rotated by $-\alpha$ when  $\mathcal{M}$ has scanned $n$
symbols $b$ in $\mathcal{T}_1$ in the process  presented above. If
$k=n$, rotations of quantum state cancel one another and the final
quantum state will be $|q_0\rangle$ with certainty. Otherwise, the
final state will be a superposition state because $\sqrt{2}$ is
irrational, and the amplitude of $|q_1\rangle$ in the state is
sufficiently large, which means that repeating the process
$\mathbf{O}(n^2)$ times guarantees getting $|q_1\rangle$ at least
once with high probability. $\mathcal{M}$ also needs to halt and
accept input $\omega\in\{a^{n}b^{n^2}\mid n\in \mathbf{N}\}$ after
repeating the process $\mathbf{O}(n^2)$ times rather than
repeating it forever. To achieve this,  We periodically execute a
subroutine that accepts with a small probability $c/n^2$, which is
much smaller than the probability of getting $|q_1\rangle$. If
$\omega\notin\{a^{n}b^{n^2}\mid n\in \mathbf{N}\}$, this does not
have much influence.

 The process is descried as follows:
\begin{description}
    \item[1] If the input $\omega=\lambda$, accept.\\
    Check whether $\omega$ is of the form $a^nb^{kn}(n>0,k>0)$. If
    not, reject.
    \item[2]Otherwise, repeat the following routine infinitum:
    \begin{description}
        \item[(1)] Move the tape head of $\mathcal{T}_1$ to\ \
        $|\hspace{-2.0mm}C_1$ and move the tape head of $\mathcal{T}_2$
        to\ \
        $|\hspace{-2.0mm}C_2$. Set the quantum state to be
        $|q_0\rangle$.
        \item[(2)] While the currently scanned symbol is not $\$_1$, do the following:
\begin{description}
    \item[1)]  If the currently scanned symbol of
        $\mathcal{T}_1$ is $a$\\
        \ \ \ \ Perform $U_{\alpha}$ on the
        quantum state and move the tape head of $\mathcal{T}_1$ one square
        right.
    \item[2)]Else if the currently scanned symbol of
        $\mathcal{T}_1$ is $b$\\
        Turn to $\mathcal{T}_2$. Move the tape head of $\mathcal{T}_2$ one square
        right and scan a symbol.
        \begin{description}
        \item[a)]If the scanned symbol of
        $\mathcal{T}_2$ is $a$\\
        Turn to $\mathcal{T}_1$, move the tape head of $\mathcal{T}_1$ one square
        right and scan a symbol.
        \item[b)]Else if the scanned symbol of
        $\mathcal{T}_2$ is $b$ \\
        Perform $U_{-\alpha}$ on the
        quantum state.
        Move the tape head of
        $\mathcal{T}_2$ back to \ $|\hspace{-2.0mm}C_2$.
       Turn to $\mathcal{T}_1$, move the tape head of $\mathcal{T}_1$ one square
        right and scan a symbol.
        \end{description}
\end{description}
        \item[(3)] Perform $U_{-\alpha}$ on the
        quantum state.
        \item[(4)] Measure the quantum state. If the
        result is $|q_1\rangle$, reject.
        \item[(5)]Repeat the following subroutine two times:
        \begin{description}
            \item[1)] Move the tape head of $\mathcal{T}_1$ back to\ \
            $|\hspace{-2.0mm}C_1$, and move the tape head of $\mathcal{T}_2$ back to \ $|\hspace{-2.0mm}C_2$.
            \item[2)] Move the tape head of $\mathcal{T}_1$ one square right and
            scan a symbol.
            \item[3)] While the currently scanned symbol is not
            \ $|\hspace{-2.0mm}C_1$ or $\$_1$, do the following:

            Simulate a coin flip.
            \begin{description}
                \item[A].  If the result is
            ``head":\\
                Move the tape head of $\mathcal{T}_1$ one square right and scan a
                symbol.
                \begin{description}
                    \item[a)]  If the symbol scanned of $\mathcal{T}_1$ is $a$, do nothing.
                    \item[b)] If the symbol scanned of $\mathcal{T}_1$ is $b$\\
                    Turn to $\mathcal{T}_2$, move the tape head of $\mathcal{T}_2$ one square right and scan a
                    symbol.
                    \begin{description}
                    \item [$\diamond$]If the scanned symbol of $\mathcal{T}_2$ is $a$\\
                    Turn to $\mathcal{T}_1$, move the tape head
                    of $\mathcal{T}_1$ one square right, scan a symbol and goto \textbf{b)}.

                    \item [$\diamond$]Else if the scanned symbol of $\mathcal{T}_2$ is
                    $b$\\
                    Move the tape head of $\mathcal{T}_2$ to
                    \
                    $|\hspace{-2.0mm}C_2$. Turn to $\mathcal{T}_1$, and move the tape head of $\mathcal{T}_1$ one square left.
                 \end{description}

                \end{description}
                \item[B]. Else\\
                 Move the tape head of $\mathcal{T}_1$ left and scan a symbol.
                 \begin{description}

                    \item[a)]  If the symbol scanned of $\mathcal{T}_1$ is $a$, do nothing.
                    \item[b)] If the symbol scanned of $\mathcal{T}_1$ is $b$\\
                    Turn to $\mathcal{T}_2$, move the tape head of $\mathcal{T}_2$ one square right and scan a
                    symbol.
                     \begin{description}
                    \item [$\diamond$]If the scanned symbol of $\mathcal{T}_2$ is $a$\\
                    Turn to $\mathcal{T}_1$, move the tape head
                    of $\mathcal{T}_1$ one square left and scan a symbol. Goto \textbf{b)}.
                   \end{description}
                     \item[c)] Move the tape head of $\mathcal{T}_2$ to
                    \
                    $|\hspace{-2.0mm}C_2$.  Turn to $\mathcal{T}_1$, and move the tape head of $\mathcal{T}_1$ one square right.
                \end{description}
            \end{description}
        \end{description}
           \item[(6)]If both times the process ends at the right
            end-marker of $\mathcal{T}_1$, simulate $l$ coin
            flips. If all the results are "heads", accept.
    \end{description}
\end{description}

A 2TQCFA $\mathcal{M}$ to recognize $L_{square}$ is defined  as
follows:
\begin{equation}
\mathcal{M}=(Q,S,\Sigma,\Theta,\delta,q_{0},s_{0},S_{acc},S_{rej})
\end{equation}

where,
\begin{itemize}
\item $Q=\{q_0, q_1\}$;

\item $S$ is a finite set of classical states;

\item $\Sigma=\{a,b\}$;\ \ $|\hspace{-2.0mm}C_1$,\
$|\hspace{-2.0mm}C_2$\  are the left end-markers of
$\mathcal{T}_1$ and $\mathcal{T}_2$, respectively; $\$_1$,\
$\$_2$\ are the right end-markers of $\mathcal{T}_1$ and
$\mathcal{T}_2$, respectively;

\item $q_{0}$ is the initial quantum state;
 \item $s_{0}$ is the
initial classical state; \item $S_{acc}=\{s_{acc}\}$;
\item$S_{rej}=\{s_{rej}\}$.
\end{itemize}

In fact, the key step is to construct quantum and classical
transition functions $\Theta$ and $\delta$. We construct them
 according to the process described
above.

For \textbf{1}: Decide whether the input $\omega$ is of the form
$a^nb^{kn}(n>0,k>0)$. For the trivial case $a^nb^{kn}(n=0)$ we
accept it at this subroutine. In this subroutine, we do not change
the quantum state. The quantum transition function is
    the identity $I$. The classical
transition function is similar to the one recognizing $w\in
\{a^{kn}b^{n}\mid n\in \mathbf{N}\}$, and we omit the details
here.

For \textbf{2-(1)}: Move the tape heads of $\mathcal{T}_1$ and
$\mathcal{T}_2$ to \ \
        $|\hspace{-2.0mm}C_1$ and \
        $|\hspace{-2.0mm}C_2$, respectively.
In this subroutine, we do not change the quantum state. The
quantum and classical transition functions are easy to describe.
We omit the details here.

 For \textbf{2-(2)} and \textbf{2-(3)}: Do the quantum operations.
\begin{itemize}
\item the classical states are
$S_2=\{s_{02},s_{a2},s_{b1},s_{re3},s_{crt2},s_{m}\}$, where
$s_{02}$ is the starting state of this subroutine; $s_{re3}$ is
the state for moving back to $\ |\hspace{-2.0mm}C_2$ of
$\mathcal{T}_2$; $s_{a2}$ and $s_{b1}$ are states to consume a
symbol $a$ and $b$ from left to right in $\mathcal{T}_2$ and
$\mathcal{T}_1$, respectively; $s_{m}$ is the ending state of this
subroutine waiting to make a measurement for the quantum state.

 \item The quantum
transition function is defined as follows:\\

  \begin{tabular}{lll}
  $\Theta(s_{02}, a)=U_{\alpha}$\ \ \ &
  $\Theta(s_{a2}, b)=U_{-\alpha}$\ \ \ \ \ \\ &
  $\Theta(s_{crt2}, b)=U_{-\alpha}$\\
  $\Theta(s, \sigma)=I$ for $s\in\{s_{b1},s_{re3}\}$&{ }&{ }
  \end{tabular}

 \item the
classical transition function is defined as follows:\\

 \begin{tabular}{ll}
    $\delta(s_{02}, a)=(s_{02}, \rightarrow,\downarrow)$&
    $\delta(s_{02}, b)=(s_{a2}, \downarrow,\rightarrow)$\\
    $\delta(s_{a2}, a)=(s_{b1}, \rightarrow,\downarrow)$&
    $\delta(s_{a2}, b)=(s_{re3}, \downarrow,\leftarrow)$\\
    $\delta(s_{re3}, a)=(s_{re3}, \downarrow,\leftarrow)$\ \ \ \ \ \ \ \ \ \ &
    $\delta(s_{re3},\ |\hspace{-2.0mm}C_2)=(s_{a2},
    \downarrow,\rightarrow)$\\
    $\delta(s_{b1}, b)=(s_{a2}, \downarrow,\rightarrow)$&
    $\delta(s_{b1}, \$_1)=(s_{crt2}, \downarrow,\rightarrow)$\\
    $\delta(s_{crt2},b)=(s_{m}, \downarrow,\leftarrow)$&{ }\\
  \end{tabular}

\end{itemize}
 For \textbf{2-(4)}: Measure the quantum state.
\begin{itemize}
\item the classical states are $S_3=\{s_{m},s_{re4}\}$.
 \item The quantum
transition function is defined as follows:\\
\begin{tabular}{l}
$\Theta(s_{m}, \sigma)=M=\{P_0,P_1\}$  for any
$\sigma\in\Sigma$,\\
\end{tabular} where
\begin{equation}
 P_0=|0\rangle\langle 0|, P_1=|1\rangle\langle 1|.
 \end {equation}
 \item the
classical transition function is defined as follows:\\
 \begin{tabular}{lll}
    $\delta(s_{m},\sigma)(1)=(s_{rej}, -,-)$,&
    $\delta(s_{m},\sigma)(0)=(s_{re4}, \downarrow,\leftarrow)$& for any $\sigma\in \Sigma$.\\
  \end{tabular}
\end{itemize}
 For \textbf{2-(5)}: Two times random walk. The random walk here is a
 little  different from the random walk in Ambainis and Watrous's paper
 \cite{AJ}. We still consider scanning a symbol $a$ as one step, but we
consider scanning  $n$ symbols $b$ as one step in this random
walk.

We do not bore all the details of quantum and classical transition
functions here. In the random walk, every time we go on a step by
scanning a symbol $a$ or $n$ symbols  $b$. Simulation of  a coin
flip is an essential component in this subroutine. The machine
$\mathcal{M}$ simulates a coin flip according to the following
transition functions, with $s_{03}$, $|q_0\rangle$ as the starting
classical and quantum states, respectively.

Let projective measurement $M=\{P_0,P_1\}$, where
\begin{equation}
 P_0=|0\rangle\langle 0|, P_1=|1\rangle\langle 1|.
 \end {equation} The
results 0 and 1 represent the results of coin flip ``head" and
``tail", respectively. Unitary operator $U$ is defined as follow:
  $$U=\left(%
\begin{array}{cc}
  \frac{1}{\sqrt{2}} &  \frac{1}{\sqrt{2}} \\
   \frac{1}{\sqrt{2}} &  -\frac{1}{\sqrt{2}} \\
\end{array}%
 \right).$$
 This operator changes the base state $|q_0\rangle$ or $|q_1\rangle$
 to be a superposition state $|\psi\rangle$ or $|\phi\rangle$, respectively, as
 follows:
\begin{equation}
|\psi\rangle=\frac{1}{\sqrt{2}}(|q_0\rangle+|q_1\rangle),
\end{equation}
\begin{equation}
|\phi\rangle=\frac{1}{\sqrt{2}}(|q_0\rangle-|q_1\rangle).
\end{equation}
When measuring $|\psi\rangle$ or $|\phi\rangle$
 with $M$, we will get the result  0 or 1 with probability of
 $\frac{1}{2}$, respectively. This is similar to the coin flip
 process. If the result is 0, we simulate a right step; if the result
 is 1, we simulate a left step.\\

 \begin{tabular}{lll}
\ \ \ \ \ \ \ \ &$\Theta(s_{03}, \sigma)=U$ &
 $\delta(s_{03}, \sigma)=(s_{03}', \downarrow, \downarrow)$\\
 {}&$\Theta(s_{03}', \sigma)=M$ &{}\\
 {}&$\delta(s_{03}', \sigma)(0)=(s_{03}, \rightarrow, \downarrow)$\ \ \ \ \ \ \ \ &
 $\delta(s_{03}', \sigma)(1)=(s_{03}, \leftarrow, \downarrow)$
 \end{tabular}\\

For \textbf{2-(6)}: Simulate $l$ coin flips. If all the results
are ``heads", the input is accepted. The classic state starts from
$s_{cf}$ and the quantum state starts from $|q_1\rangle$. We still
use unitary operator $U$ and projective measurement $M$ described
as above to simulate a coin flip. We use the measurement result 0
and 1 to represent ``head" and  ``tail", respectively.
\begin{itemize}
\item the classical states are
$\{s_{fail},s_{acc},s_{cf}^{(0,i)},s_{cf}^{(1,i)},
i=0,1,2,...,l+1\}$, where $s_{cf}^{(0,0)}$ is the starting state.

 \item The quantum and classical transition functions are defined as follows:\\
 for $0\leq i\leq l$,\\

 \begin{tabular}{ll}
   $\Theta(s_{cf}^{(0,i)},\$_1)=U$ & $\delta(s_{cf}^{(0,i)},\$_1)=(s_{cf}^{(1,i)},\downarrow,\downarrow)$  \\
   $\Theta(s_{cf}^{(1,i)},\$_1)=M$ & {}\\
   $\delta(s_{cf}^{(1,i)},\$_1)(1)=(s_{fail},\leftarrow,\downarrow)$\ \ \ \ \ \  & $\delta(s_{cf}^{(1,i)},\$_1)(0)=(s_{cf}^{(0,i+1)},\downarrow,\downarrow)$\\
 \end{tabular}\\

 When the classical state changes to $s_{fail}$, it means that the
 ``tail" result of a coin flip has occurred. There is no need to simulate
 any more coin flips. The tape head of $\mathcal{T}_1$ should be moved back to the first
 symbol of the tape and  another new iteration should be started. We omit the
 transition function here. When the classical state changes to
 $s_{cf}^{(0,l+1)}$, it means that all the results of the $k$ coin flips
 are ``heads". The classical state should be changed to an accepting
 state. The transition function is as follow:\\

\begin{tabular}{l}
   $\delta(s_{cf}^{(0,l+1)},\$_1)=(s_{acc},-,-)$.\\
 \end{tabular}\\

\end{itemize}

\begin{Lm}
If the input is $\omega=a^nb^{kn}$ and $k\neq n$, $\mathcal{M}$
rejects after \textbf{2-(2),2-(3),2-(4)} with probability at least
$1/2(n-k)^2$.
\end{Lm}
\begin{proof}
The state $|q_0\rangle$ is rotated by $\sqrt{2}\pi$ when
$\mathcal{T}_1$ scans one symbol $a$, and it is rotated by
$-\sqrt{2}\pi$ when $\mathcal{T}_1$ scans $n$ symbols $b$. The
$|q_0\rangle$ is rotated by $\sqrt{2}(n-k)\pi$ after
\textbf{2-(2),2-(3)}. So the quantum state of $\mathcal{M}$ after
rotating is
$\cos(\sqrt{2}(n-k)\pi)|q_0\rangle+\sin(\sqrt{2}(n-k)\pi)|q_1\rangle$.
In \textbf{2-(4)} the probability of observing  $|q_1\rangle$ is
$\sin^2(\sqrt{2}(n-k)\pi)$. Without loss of generality, we assume
$n-k>0$. Let $l$ be the closest integer to $\sqrt{2}(n-k)$. Assume
that $\sqrt{2}(n-k)>l$ (the other case is symmetric), then
$2(n-k)^2>l^2$. So we get $2(n-k)^2-1\geq l^2$ and $l\leq
\sqrt{2(n-k)^2-1}$. We have
\begin{equation}
\sqrt{2}(n-k)-l\geq \sqrt{2}(n-k)-\sqrt{2(n-k)^2-1}
\end{equation}

\begin{equation}
=\frac{(\sqrt{2}(n-k)-\sqrt{2(n-k)^2-1})(\sqrt{2}(n-k)+\sqrt{2(n-k)^2-1})}{\sqrt{2}(n-k)+\sqrt{2(n-k)^2-1}}
\end{equation}
\begin{equation}
=\frac{1}{\sqrt{2}(n-k)+\sqrt{2(n-k)^2-1}}>\frac{1}{2\sqrt{2}(n-k)}.
\end{equation}

Because $l$ is the closest integer to $\sqrt{2}(n-k)$, we have
$0<\sqrt{2}(n-k)-l<1/2$. Assume $f(x)=sin(x\pi)-2x$. We have
 $f''(x)=-\pi^2\sin(x\pi)\leq 0$ when $x\in
[0,1/2]$. That is to say, $f(x)$ is concave in $[0,1/2]$, and we
have $f(0)=f(1/2)=0$. So for any $x\in[0,1/2]$, it holds that
$f(x)\geq 0$, that is, $\sin(x\pi)\geq 2x$. Therefore
\begin{equation}
\sin^2(\sqrt{2}(n-k)\pi)=\sin^2((\sqrt{2}(n-k)-l)\pi)
\end{equation}
\begin{equation}
\geq (2(\sqrt{2}(n-k)-l))^2 =4(\sqrt{2}(n-k)-l)^2
\end{equation}
\begin{equation}
\geq 4(\frac{1}{2\sqrt{2}(n-k)})^2=\frac{1}{2(n-k)^2}.
\end{equation}
So the lemma has been proved.
\end{proof}

\begin{Lm}
If the input is $\omega=a^nb^{kn}$ and $k\neq n$, $\mathcal{M}$
accepts after \textbf{2-(5),2-(6)} with probability
$1/2^l(n+k+1)^2$.
\end{Lm}
\begin{proof}
\textbf{2-(5)} is similar to two times of random walk (takes
scanning $n$ symbols $b$ as one step) starting at location 1 and
ending at location 0 (the left end-marker\ \
$|\hspace{-2.0mm}C_1)$) or at location $(k+1)n+1$ (the right
end-marker \$). It can be known from probability theory that the
probability of reaching the location $(k+1)n+1$ is $1/(n+k+1)$
(see Chapter14.2 in \cite{WF}). Repeating it twice and flipping
$l$ coins, we  get the probability $1/2^l(n+k+1)^2$.
\end{proof}

Let $l=1+  \lceil\log \varepsilon\rceil$. If
$\omega=a^nb^m$,$m\neq kn$, then step \textbf{1} always rejects
it. If $\omega=a^nb^{n^2}$, \textbf{2-(2),2-(3),2-(4)} always turn
$|q_0\rangle$ to $|q_0\rangle$, and $\mathcal{M}$ never rejects.
After \textbf{2-(5),2-(6)}, the probability of $\mathcal{M}$
accepting $\omega$ is $1/2^l(n+k+1)^2$. Repeating step \textbf{2}
for $cn^2$ times, the accepting probability is
$1-(1-\frac{1}{2^l(n+k+1)^2})^{cn^2}$, and this can be made
arbitrarily close to 1 by selecting constant $c$ appropriately. If
$x=a^nb^{kn}$ and $k\neq n$, $\mathcal{M}$ rejects after
\textbf{2-(2),2-(3),2-(4)} with probability $P_r>1/2(n-k)^2$.
$\mathcal{M}$ accepts after \textbf{2-(5),2-(6)} with probability
$P_a=1/2^l(n+k+1)^2\leq \varepsilon/2(n+k+1)^2$. If we repeat step
\textbf{2} indefinitely, the probability of rejecting is
\begin{equation}
\sum_{i\geq 0}(1-P_a)^i(1-P_r)^iP_r=\frac{P_r}{P_a+P_r-P_aP_r}
\end{equation}
\begin{equation}
>\frac{P_r}{P_a+P_r}>\frac{1/2}{1/2+\varepsilon/2}=\frac{1}{1+\varepsilon}>1-\varepsilon.
\end{equation}

If we assume the input $\omega=a^nb^m$, step \textbf{1} takes
$\mathbf{O}(n+m)$ time, step \textbf{2-{1}} takes
$\mathbf{O}(n+m)$ time, \textbf{2-{2},2-{3},2-{4}} take
$\mathbf{O}(n+m)$ time, and \textbf{2-{5},2-{6}} take
$\mathbf{O}(n+m)^2$ time. The expected number of repeating step
\textbf{2} is $\mathbf{O}((n+m/n)^2)$. Hence, the expected running
time of $\mathcal{M}$ is at most $\mathbf{O}((n+m/n)^2(n+m)^2)$.
\end{proof}

\section{Languages recognized by $m$TFA and $m$TQCFA}
We know that 2QCFA are more powerful than 2DFA, but we do not know
whether $k$TQCFA are more powerful than $k$TFA or not. Although we
can not prove $k$TQCFA are strictly more powerful than $k$TFA, we
will give an example which seems to imply that the result may be
true. We will show that $ \{a^{n}b^{n^{k}}\mid n\in
\mathbf{N}\}(k=1,2...)$ can be recognized by $(k+1)$TFA in linear
time. Also, we will prove that $ \{a^{n}b^{n^{k}}\mid n\in
\mathbf{N}\}(k=1,2...)$ can be recognized by $k$TQCFA with one
side-error in polynomial time. Thus, it seems that $k$TQCFA are
more powerful than $k$TFA.

\begin{Th}
$ L_{square}=\{a^{n}b^{n^{2}}\mid n\in \mathbf{N}\}$ can be
recognized by 3TFA in linear time.
\end{Th}

\begin{proof}

The process to recognize $L_{square}$ is described as follows:
\begin{description}
\item[1] If the input $\omega=\lambda$, accept. \item[2] Check
whether the input $\omega$ is of the form $a^+b^+$. If not,
reject. \item[3] Move the tape head of $\mathcal{T}_1$ to the
first symbol $b$ in $\mathcal{T}_1$.
\item[4] While the currently scanned symbol by $\mathcal{T}_1$ is $b$, do the following:\\
 Turn to $\mathcal{T}_2$, move the tape head of $\mathcal{T}_2$ one square right and scan a symbol.
\begin{description}
\item[(1)] If the currently scanned symbol by $\mathcal{T}_2$ is $b$\\
Move the tape head of $\mathcal{T}_2$ back to \
$|\hspace{-2.0mm}C_2$. Turn to $\mathcal{T}_3$, move the tape head
of $\mathcal{T}_3$ one square right and scan a symbol.
\begin{description}
\item[(A)] If the currently scanned symbol by $\mathcal{T}_3$ is $b$\\
Reject.
\item[(B)] Else\\
 Turn to $\mathcal{T}_1$, move the tape head of $\mathcal{T}_1$ one square right and scan a symbol.
\end{description}
\item[(2)] Else\\
  Turn to $\mathcal{T}_1$, move the tape head of $\mathcal{T}_1$ one square right and scan a symbol.
\end{description}
\item[5]  Turn to $\mathcal{T}_2$, move the tape head of
$\mathcal{T}_2$ one square right and scan a symbol.
\begin{description}
\item[(1)]If the currently scanned symbol by $\mathcal{T}_2$ is not $b$\\
Reject.
\item[(2)] Else\\
Turn to $\mathcal{T}_3$, move the tape head of $\mathcal{T}_3$ one
square right and scan a symbol.
\begin{description}
\item[(A)] If he currently scanned symbol by $\mathcal{T}_2$ is not $b$\\
Reject. \item[(B)] Else, accept.
\end{description}
\end{description}
\end{description}
We assume the input $\omega=a^{n}b^m$. Let $k_2$ be the number of
symbol $a$ having been scanned by tape head $\mathcal{T}_2$ and
$k_3$  be the number of symbol $a$ having been scanned by tape
head $\mathcal{T}_3$  from left to right. Actually,  $k_2k_3$ is
acting
 like a $2$ digits of $n$-radix number. When the tape head of $\mathcal{T}_1$  scans a symbol $b$, we just add $1$ to $k_2$.
When $k_2$ is large enough ($k_2=n$), we turn $k_2$ to 0 and add
$1$ to $k_3$. It is easy to verify that this process will accept $
L_{square}=\{a^{n}b^{n^{2}}\mid n\in \mathbf{N}\}$ in linear time.

\end{proof}

The above result is easy to extend to the case of $(k+1)$TFA. So
we get the following Theorem.

\begin{Th}
$ \{a^{n}b^{n^{k}}\mid n\in \mathbf{N}\}(k=1,2,...)$ can be
recognized by $(k+1)$TFA in linear time.
\end{Th}
\begin{proof} The idea to prove this theorem is almost the same as the  one   in the previous theorem.
 Let $\mathcal{T}_1$ scan
symbol $b$, and the numbers of symbols $a$ having been scanned by
$\mathcal{T}_i(i=2,3,...,k+1)$ are acting like a $k$ digits of
$n$-radix number. The details are similar to the above theorem,
and thus we do not bore here.
\end{proof}

\begin{Co}\label{co1}
$\{a^{n}b^{rn^{k-1}}\mid n\in \mathbf{N},r\in
\mathbf{N}\}(k=1,2,...)$ can be recognized by $k$TFA in linear
time.
\end{Co}
\begin{proof} Let $n_k$ be the number of symbol $a$ having been scanned by
$\mathcal{T}_k$.  When $n_k$ reaches $n$, and if there are
 symbols $b$ in $\mathcal{T}_1$ which have not been scanned, we do not reject the input, but
return tape head of  $\mathcal{T}_k$ to \ $|\hspace{-2.0mm}C_k$
and go on. The other details are the same as $k$TFA recognizing
language $\{a^{n}b^{n^{k-1}}\mid n\in \mathbf{N}\}$.
\end{proof}

\begin{Th}
For any $\epsilon>0$, there is a $k$TQCFA $\mathcal{M}$ that
accepts any $\omega \in \{a^{n}b^{n^{k}}\mid n\in \mathbf{N}\}$
with certainly, rejects any $\omega \notin \{a^{n}b^{n^{k}}\mid
n\in \mathbf{N}\}$ with probability at least $1-\epsilon$ and
halts in polynomial time.
\end{Th}
\begin{proof} In terms of Colollary~\ref{co1}, we can
check whether $w$ is of the form  $a^{n}b^{rn^{k-1}}$ in linear
time. Then we use the similar method as in Theorem~\ref{th1} to
verify $r=n$.
 The process is similar to that in Theorem~\ref{th1}, except the random
walk. $k$TQCFA will take scanning $n^{k-1}$ symbols $b$ as one
step in the random walk. We omit the details here.
\end{proof}

\section{Conclusions  }
In this paper, we have studied 2TFA and proved several languages
can be recognized by 2TFA. Augmenting with a  quantum component of
constant size, we have proposed a new computing model called
2TQCFA, and we have proved that $L_{square}=\{a^{n}b^{n^{2}}\mid
n\in \mathbf{N}\}$ can be recognized by 2TQCFA in polynomial time
with one-sided error. Furthermore, we have proposed $m$TFA and
$M$TQCFA and proved that $\{a^{n}b^{n^{k}}\mid n\in \mathbf{N}\}$
can be recognized by $k$TQCFA.

 As we know, 2QCFA are more
powerful than 2DFA. Thus, in the future we would like to consider
this question: Are 2TQCFA more powerful than 2TFA? Furthermore,
are kTQCFA more powerful than kTFA?  We have proved that all
languages which have been shown to be recognized by 2QCFA or 2QFA
can also be recognized by 2TFA in this paper. It is natural to ask
whether or not there is any language which can be recognized by
2QCFA or 2QFA but can not be recognized by 2TFA?

\subsubsection*{Acknowledgments}
This work is supported in
part by the National Natural Science Foundation (Nos.
60873055, 61073054), the
Natural Science Foundation of Guangdong Province of China (No. 10251027501000004), the Fundamental Research Funds for the Central Universities (Nos. 10lgzd12,11lgpy36), the Research
Foundation for the Doctoral Program of Higher School of Ministry
of Education (Nos. 20100171110042, 20100171120051),  the China Postdoctoral Science Foundation project (Nos. 20090460808, 201003375),
and the project of  SQIG at IT, funded by FCT and EU FEDER projects
Quantlog
POCI/MAT/55796/2004 and
QSec PTDC/EIA/67661/2006, IT Project QuantTel, NoE Euro-NF, and the SQIG LAP initiative.

\nocite{*}
\bibliographystyle{fundam}

\end{document}